\documentclass[11pt]{article}

\usepackage[letterpaper]{geometry}
\usepackage[dvipdfmx]{color}
\usepackage[dvipdfmx]{graphicx}
\usepackage{tikz}
\usepackage{gnuplot-lua-tikz}
\usepackage[utf8]{inputenc} 
\usepackage[T1]{fontenc}
\usepackage{url}
\usepackage{ifthen}
\usepackage{cite}
\usepackage[cmex10]{amsmath}
\usepackage{amssymb,amsthm}
\usepackage{mathtools}
\newtheorem{theorem}{Theorem}
\newtheorem{corollary}{Corollary}
\newtheorem{lemma}{Lemma}

\newcommand{\ignore}[1]{}

\newcommand{\F}{\mathbb{F}}

\newcommand{\ldis}{{d}_\mathrm{L}}

\newcommand{\D}{\mathsf{D}}
\newcommand{\I}{\mathsf{I}}

\begin{document}

\title{Improved Bounds for Codes \\Correcting Insertions and Deletions}

\author{
    Kenji Yasunaga\thanks{Tokyo Institute of Technology, Japan. E-mail: \texttt{yasunaga@c.titech.ac.jp}}
}

\maketitle

\begin{abstract}
  This paper studies the cardinality of codes correcting insertions and deletions.
  We give improved upper and lower bounds on code size.
  Our upper bound is obtained by utilizing the asymmetric property of list decoding for insertions and deletions
  and can be seen as analogous to the Elias bound in the Hamming metric.
  Our non-asymptotic bound is better than the existing bounds when the minimum Levenshtein distance is relatively large.
  The asymptotic bound exceeds the Elias and the MRRW bounds adapted from the Hamming-metric bounds for the binary and the quaternary cases.
  Our lower bound improves on the bound by Levenshtein, but its effect is limited and vanishes asymptotically.
\end{abstract}

\section{Introduction}

We study the existence of codes correcting insertions and deletions.
We are interested in deriving upper and lower bounds on the cardinality of codes.

Levenshtein~\cite{Lev66} gave asymptotic upper and lower bounds for codes
correcting a constant number of insertions and deletions.
Later, he presented bounds for correcting any number~\cite{Lev02}.
Following his work, there have been several studies~\cite{KMTU11,KK13,CK14} on the cardinality of codes.
However, they mainly focused on codes correcting a constant number of insertions/deletions.
Giving better bounds for codes correcting a constant \emph{fraction} of insertions/deletions has been elusive.
See~\cite{CR21} for a recent survey.

In this work, we present upper and lower bounds on the cardinality of codes correcting insertions and deletions.
First, we give a non-asymptotic upper bound on the cardinality of codes correcting insertions/deletions.
Asymptotically, it implies that for any code $C \subseteq \Sigma^n$ of rate $R$ that can correct $\lfloor \delta n \rfloor - 1$ insertions/deletions,
it holds that  $R \leq (1 - H_q(\delta))/(1-\delta)$,
where $|\Sigma| = q$, $\delta \in [0,1)$, and $H_q(\cdot)$ is the $q$-ary entropy function.
The bound improves on the asymptotic upper bounds from the literature; 
The well-known Elias and MRRW bounds for the Hamming metric can also be employed as upper bounds in the Levenshtein metric (insertions and deletions).
Our bound improves on them for $q = 2, 4$.

Our bound is obtained by a similar argument  to the Elias bound in the Hamming metric.
We use the list size upper bound of~\cite{HY20} for insertions and deletions.
It is well-known that any $s$-deletion correcting code  can correct any $s_1$ insertions and $s_2$ deletions with $s_1+s_2=s$.
This symmetry in the unique decoding regime does not hold in the list decoding problem.
In~\cite{HY20}, it is proved that any code with a large Levenshtein distance enables list decoding
such that  the decoding radius of insertion is superior to that of deletion.
We crucially use this property to derive our upper bound.

Next, we give a non-asymptotic lower bound on the cardinality of codes.
Our bound improves on the bound by Levenshtein~\cite{Lev02} by investigating that every deletion ball contains multiple words
that are close to each other in the Levenshtein distance.
Asymptotically, our lower bound is the same as in~\cite{Lev02}.

Finally, we compare our bounds with the existing bounds in the literature.
As a non-asymptotic bound, our upper bound is tighter than the bounds in~\cite{Lev02,KK13} when the minimum distance is relatively large.
Our asymptotic upper bound is the best compared with the existing bounds from the Hamming metric for the binary and the quaternary codes.
Regarding lower bounds, although our bound strictly improves the bound in~\cite{Lev02}, its effect is limited and vanishes asymptotically.

\subsection*{Related Work}

The existence of positive rate codes correcting a constant fraction of insertions/deletions was presented by Schulman and Zuckerman~\cite{SZ99}.
Their asymptotic lower bounds  were later improved by Levenshtein~\cite{Lev02}.

Cullina and Kiyavash~\cite{CK14} improved Levenshtein's upper bound~\cite{Lev02} for correcting a constant number of insertions/deletions
using a graph-theoretic approach.
Kulkarni and Kiyavash~\cite{KK13} derived non-asymptotic upper bounds by a linear programming argument for graph-matching problems.
They also gave upper bounds for correcting a constant fraction of insertions/deletions.
Although their asymptotic bound (rate function) improved on the bound in~\cite{Lev02}, it was not given in the closed-form expressions.

For extreme cases where the deletion fraction is either small or high,
Guruswami and Wang~\cite{GW17} gave efficient constructions of codes correcting these cases of deletions.
For the case that the coding rate is nearly zero,
Kash et al.~\cite{KMTU11} showed 
a positive-rate binary code correcting the fraction $p$ of insertions/deletions with $p \geq 0.1737$,
which improved on the bound of $p \geq 0.1334$ in~\cite{Lev02}.
Bukh, Guruswami, and Hast\r{a}d~\cite{BGH17} significantly improved on the previous results by showing
the existence of a positive rate $q$-ary code with $p \geq 1 - (2/(q+\sqrt{q}))$,
which is $p \geq 0.4142$ for the binary case.
Guruswami, He, and Li~\cite{GHL21} showed a slight but highly non-trivial improvement
to the upper bound on the fraction of correctable insertions/deletions by codes with a non-zero coding rate.
They proved that there is a constant $\epsilon >0$ such that any $q$-ary code correcting a $(1 - (1+\epsilon)/q)$-fraction of
insertions/deletions must have a rate approaching $0$.

Maringer et al.~\cite{MPVW21} studied insertion/deletion correcting codes in feedback models
and completely determined the maximal asymptotic rate for binary codes.

\section{Code Size Upper Bound}

Let $\Sigma$ be a finite alphabet. 
The Levenshtein distance $\ldis(x, y)$ between two words $x$ and $y$ is
the minimum number of symbol insertions and deletions needed to transform $x$ into $y$.
For a code $C \subseteq \Sigma^n$, its minimum Levenshtein distance is
 the minimum distance $\ldis(c_1, c_2)$ of every pair of distinct codewords $c_1, c_2 \in C$.
 Since any two codewords in $C$ are of the same length, the minimum Levenshtein distance of $C$ takes an even number.
 It is well-known that a code with minimum Levenshtein distance $d$ can correct any $s_1$ insertions and $s_2$ deletions
 as long as $s_1 + s_2 \leq d/2 - 1$.
 The Levenshtein distance between two words in $\Sigma^n$ takes integer values from $0$ to $2n$.
 Thus, we consider the normalized Levenshtein distance $\delta = d/2n$ in the analysis.
 The value $\delta \in [0,1]$ also represents the fraction of insertions/deletions that can be corrected
 since we require $(s_1+s_2)/n \leq (d/2 - 1)/n = \delta - 1/n$, which is asymptotically equal to $\delta$.

Let $C \subseteq \Sigma^n$ be a code of minimum Levenshtein distance $d$ with $|\Sigma|=q$.
For a word $x \in \Sigma^n$,
let $I_t(x) \subseteq \Sigma^{n+t}$ be the set of its supersequences of length $n+t$.
Namely, $I_t(x)$ consists of words obtained from $x$ by inserting $t$ symbols.
Similarly, let $D_t(x)$ be the set of words obtained from $x$ by deleting $t$ symbols.
It is known~\cite{Lev01} that the size of $I_t(x)$ does not depend on $x$ and 
\[|I_t(x)| = \sum_{i=0}^t \binom{n+t}{i}(q-1)^i \triangleq I_q(n,t).\]

First, we give a simple sphere packing bound.
We use the fact that the number of supersequences, $|I_t(x)|$, is independent of $x$.

\begin{theorem}\label{thm:spbound}
  Let $C \subseteq \Sigma^n$ be a code of minimum Levenshtein distance $d$ and $|\Sigma|=q$.
  It holds that
  \begin{equation}
    |C| \leq \left\lfloor\frac{q^{n+d/2-1}}{I_q(n,d/2-1)}\right\rfloor.
    \end{equation}
\end{theorem}
\begin{proof}
  For each codeword $c \in C$, consider the set of supersequences $I_{d/2-1}(c) \subseteq \Sigma^{n+d/2-1}$.
  Since the code has minimum Levenshtein distance $d$, each $I_{d/2-1}(c)$ should be disjoint. Thus,
  \[ \sum_{c \in C} |I_{d/2-1}(c)| \leq q^{n+d/2-1}.\]
  The statement follows by the equality $|I_{d/2-1}(c)| = I_q(n, d/2-1)$ for any $c \in C$.
\end{proof}

Next, we prove our main theorem,
which can be seen as an Elias-type upper bound on the code size in the Levenshtein metric.

\begin{theorem}\label{thm:bound}
  Let $C \subseteq \Sigma^n$ be a code of minimum Levenshtein distance $d < 2n$ and $|\Sigma|=q$.
  For any integer $t \geq 0$ with
  \begin{equation}
    t < \frac{nd}{2n-d},\label{eq:cond}
  \end{equation}
  it holds that
  \begin{equation}
    |C|\leq \left\lfloor\frac{(n+t)d}{(n+t)d - 2nt}\cdot \frac{q^{n+t}}{I_q(n,t)}\right\rfloor. \label{eq:bound}
    \end{equation}
\end{theorem}
\begin{proof}
  By double counting, 
 the sum of the cardinalities of $I_t(x)$ of all $x \in \Sigma^n$ is
 equal to the sum of the cardinalities of $D_t(y)$ of all $y \in \Sigma^{n+t}$.
 Namely,
\begin{align}
  \sum_{y \in \Sigma^{n+t}}|D_t(y)| = \sum_{x \in \Sigma^n}|I_t(x)| = q^n \cdot I_q(n,t).\label{eq:leveq}
  \end{align}
By considering the intersection with $C$,
\[ \sum_{y \in \Sigma^{n+t}}|D_t(y) \cap C| = \sum_{x \in C}|I_t(x)| = |C| \cdot I_q(n,t).\]

Thus, by choosing $y \in \Sigma^{n+t}$ uniformly at random, 
\begin{align*}
  \mathbb{E}_{y}[ |D_t(y) \cap C|] & = \frac{1}{q^{n+t}} \sum_{y \in \Sigma^{n+t}} |D_t(y) \cap C|
   = \frac{|C|\cdot I_q(n,t)}{q^{n+t}}.
\end{align*}
The averaging argument implies that there exists $y \in \Sigma^{n+t}$ such that
\begin{equation}
  |D_t(y) \cap C| \geq \frac{|C|\cdot I_q(n,t)}{q^{n+t}}.\label{eq:ave}
\end{equation}

We have the following lemma.
\begin{lemma}\label{lem:hy20}
  For any non-negative integer $t$ with $t < nd/(2n-d)$ and $y \in \Sigma^{n+t}$,
  it holds that
  \[|D_t(y) \cap C| \leq \frac{(n+t)d}{(n+t)d - 2nt}.\]
\end{lemma}
\begin{proof}
  In \cite{HY20}, $B_\mathsf{L}(v, t_\I, t_\D)$ is defined as the set of words that can be obtained from $v$
  by at most $t_\I$ insertions and at most $t_\D$ deletions.
  Then, for $C \subseteq \Sigma^n$ and $y \in \Sigma^{n+t}$,
  we can see that $B_\mathsf{L}(y, 0, t) \cap C = D_t(y) \cap C$.
  Lemma~1 of \cite{HY20} with $t_\I=t$ and $t_\D=0$ implies that
  if $d/2 > t n/(n+t)$,
  $|D_t(y) \cap C| \leq (n+t)(d/2)/((n+t)d/2 - tn)$ for every $y \in \Sigma^{n+t}$,
  which implies the statement.
\end{proof}
By combining (\ref{eq:ave}) and Lemma~\ref{lem:hy20}, the statement follows.
\end{proof}

We analyze asymptotics of Theorems~\ref{thm:spbound} and~\ref{thm:bound}.
For a code $C \subseteq \Sigma^n$ of distance $d$, let $\delta = d/2n$ and $\gamma = t/n$ for $t \geq 0$.
Let $A_q(n,d)$ be the maximum size of code $C \subseteq \Sigma^n$ of minimum Levenshtein distance $d$, where $|\Sigma| = q$.
For  $\delta \in [0,1]$, let
\[ R_q(\delta) \triangleq \liminf_{n \to \infty}\frac{\log_q A_q(n,\lfloor 2\delta n \rfloor)}{n}\]
be the asymptotic coding rate achievable for normalized  Levenshtein distance $\delta$.
Note that $R_q(\delta) = 0$ for $\delta \geq 1-q^{-1}$ (See, for example, \cite[Section~1]{BGH17}).

Let $\mathrm{Vol}_q(n, \ell)$ be the volume of the Hamming ball of radius $\ell$ in $\F_q^n$.
Namely, $\mathrm{Vol}_q(n,\ell) = \sum_{i=0}^\ell \binom{n}{i}(q-1)^i$.
It is well-known (cf.~\cite[Lemma 4.8]{Roth06}) that, for $0 \leq \ell \leq n$,
\[\mathrm{Vol}_q(n, \ell) \geq \frac{1}{n+1}\cdot q^{nH_q(\ell/n)},\]
where $H_q(x) = -x \log_qx - (1-x)\log_q(1-x)+x\log_q(q-1)$.
Since $I_q(n, t) = \mathrm{Vol}_q(n+t, t)$, 
\[\frac{1}{n} \cdot \log_q I_q(n,t) \geq (1+\gamma)H_q\left(\frac{\gamma}{1+\gamma}\right) - \frac{\log_q((1+\gamma)n+1)}{n}.\]

Regarding Theorem~\ref{thm:spbound}, it holds that
\[ \frac{1}{n} \cdot \log_q I_q(n, d/2-1) \geq (1+\delta) \cdot H_q\left(\frac{\delta}{1+\delta} - o(1)\right) - o(1). \]
By Theorem~\ref{thm:spbound}, the rate $R$ of $C$ is bounded above by
\begin{align*}
R = \frac{\log_q |C|}{n} & \leq (1+\delta)\left( 1 - H_q\left(\frac{\delta}{1+\delta} - o(1) \right)\right) + o(1).
\end{align*}

Thus, we have the following corollary.
\begin{corollary}\label{cor:simple_ub}
  \[ R_q(\delta) \leq (1+\delta)\left( 1 - H_q\left(\frac{\delta}{1+\delta}\right)\right).\]
\end{corollary}

Regarding Theorem~\ref{thm:bound},
condition (\ref{eq:cond}) 
can be rewritten as
\begin{equation}
  \gamma < \frac{\delta}{1-\delta}.
\end{equation}

Let $\gamma = \delta/(1-\delta) - 1/n$. 
The bound (\ref{eq:bound}) can be rewritten as
\begin{align*}
  |C| & \leq \frac{(1+\gamma)\delta}{(1+\gamma)\delta - \gamma}\cdot \frac{q^{(1+\gamma)n}}{I_q(n,t)}\\
  & = \frac{\delta/(1-\delta) - \delta/n}{(1-\delta)/n} \cdot \frac{q^{(1+\gamma)n}}{I_q(n,t)}\\
  & = q^{(1+\gamma)n\left(1 - H_q\left(\frac{\gamma}{1+\gamma}\right)+o(1)\right)}.
\end{align*}
Since
\[ \frac{\gamma}{1+\gamma} = \frac{\delta/(1-\delta)-1/n}{1/(1-\delta)-1/n} = \delta - \frac{(1-\delta)^2}{n-(1-\delta)},\]
the rate $R$ of $C$ is bounded above by
\begin{align*}
  R = \frac{\log_q |C|}{n} & \leq (1+\gamma) \left(1 - H_q\left(\frac{\gamma}{1+\gamma}\right)\right) + o(1)\\
  & \leq \frac{1}{1-\delta}\left(1 - H_q\left(\delta - o(1)\right)\right) + o(1).
\end{align*}
We obtain the following corollary.
\begin{corollary}\label{cor:ub}
  \[
  R_q(\delta) \leq \frac{1}{1-\delta} \left( 1 - H_q(\delta)\right)
  \]
\end{corollary}

\subsection*{Bounds from Hamming-Metric Bounds}

For a code $C \subseteq \Sigma^n$, let $d_h$ be the minimum Hamming distance of $C$.
It is well-known that the minimum Levenshtein distance $d$ of $C$ satisfies $d \leq 2d_h$.
Thus, whenever we have an upper bound for a code of minimum Hamming distance $d_h$,
we can use the same bound for a code of minimum Levenshtein distance $d_h/2$.
More specifically, if a code $C \subseteq \Sigma^n$ satisfies $|C| \leq f(q,n,d_h)$, then 
$|C| \leq f(q,n, d/2)$ as long as $f$ is a monotonically decreasing function on the third argument.
Let $\delta = d/2n$ and $\delta_h = d_h/n$.
Similarly, if we have a bound on the coding rate $R \leq g(q,\delta_h)$,
then we have $R_q(\delta) \leq g(q,\delta)$ if $g$ is monotonically decreasing on the second argument.

As far as we know, the best-known upper bounds for coding rate with respect to normalized minimum Levenshtein distance are
obtained by the bounds for Hamming metric.
The following bounds are well-known in the literature.
\begin{theorem}
  For $0 \leq \delta < \theta = 1 - q^{-1}$,
  \begin{align*}
    R_q(\delta)  & \leq 1 - H_q\left( \theta - \sqrt{\theta(\theta - \delta)}\right); & \text{(Elias bound)}\\
    R_q(\delta) & \leq H_q\left( \frac{1}{q}\left( q - 1 - (q-2)\delta - 2\sqrt{\delta(1-\delta)(q-1)}\right)\right). & \text{(MRRW bound)}
  \end{align*}
\end{theorem}

\section{Code Size Lower Bounds}

Next, we consider lower bounds on $A_q(n,d)$ for codes $C \subseteq \Sigma^n$ with $|\Sigma|=q$.
We assume that $d$ is even.
For $x \in \Sigma^n$ and  non-negative integers $t, s$ with $t \leq n$,
let $L_{t, s}(x)$ be the set of words that can be obtained from $x$ by deleting $t$ symbols and inserting $s$ symbols.
By definition, it holds that $L_{t,0}(x) = D_t(x)$ and $L_{0,s}(x) = I_s(x)$.

We would like to derive an upper bound on the average size of $L_{t, t}(x)$ for $x \in \Sigma^n$ and $t \leq n$.
This is because it gives a lower bound on $A_{q}(n,d)$ as discussed in \cite{Tol97,Lev02}.
More specifically, for any $X \subseteq \Sigma^n$, let
\[ V_t^\mathrm{ave}(X) \triangleq \frac{1}{|X|} \sum_{x \in X} |L_{t,t}(x)| \]
be the average size of $L_{t,t}(x)$ in $X$.
Then, there exists a code $C \subseteq \Sigma^n$ of minimum Levenshtein distance $d$ satisfying
\begin{align}
  |C| \geq \frac{|X|}{V_{d/2 - 1}^\mathrm{ave}(X)} = \frac{|X|^2}{\sum_{x \in X} |L_{d/2-1,d/2-1}(x)|}.\label{eq:avelower}
\end{align}

For $x = (x_1, \dots, x_n) \in \Sigma^n$ and $i \in [1, n-1]$,
we say $(x_i,x_{i+1})$ is a \emph{distinct adjacent pair} if $x_i \neq x_{i+1}$.
Let $IP(x) = \{ (i, {i+1}) : x_i \neq x_{i+1} \}$ be the index-pair set of distinct adjacent pairs in $x$.
We say $P \subseteq IP(x)$ is \emph{disjoint} if
$P = \{({i_1},{i_1+1}), \dots, ({i_p}, {i_p+1})\}$
satisfies $i_j+1 \neq i_\ell$ for every distinct $j, \ell \in \{1, \dots, p\}$.
We denote by $p(x)$ the maximum size of disjoint index-pair sets  of distinct adjacent pairs in $x$.
For example, for $x = 01101$, we have $IP(x) = \{ (1,2), (3,4), (4,5)\}$.
The sets $\{(1,2), (3,4)\}$ and $\{(1,2), (4,5)\}$ are disjoint, but $\{(3,4),(4,5)\}$ is not.
We have $p(01101) = 2$.
For consistency, we consider the \emph{leftmost} maximum-sized disjoint index-pair set in $x$.
In the case that $x = 01101$, the leftmost one is $\{(1,2), (3,4)\}$.
Similarly, for $y = 010100$ and $y' = 010111$, the (leftmost) maximum-sized disjoint index-pair set is $\{ (1,2), (3,4)\}$,
and hence $p(y) = p(y') = 2$.

We observe that for any positive integer $p' \leq p(x)$ satisfying $p' \leq \min\{t, n-t\}$,
there are $2^{p'}$ words in $D_t(x)$ such that they are within a Levenshtein distance of $2p'$ from each other.
The reason is as follows.
Since $p' \leq p(x)$, there are $n - 2p'$ indices in $\{1, \dots, n\}$ such that
they are not contained in a disjoint index-pair set of distinct adjacent pairs in $x$.
First, we delete $t - p'$ symbols from $x$ out of the $n - 2p'$ indices.
This procedure requires that $p' \leq \min\{t,n-t\}$.
Let $y \in \Sigma^{n-t+p'}$ be the resulting word.
Second, we delete one of the two symbols in $p'$ distinct adjacent pairs from $y$.
There are $2^{p'}$ possible deletion patterns,
resulting in all different words of length $n-t$.
Since each resulting word and $y$ have Levenshtein distance $p'$,
the $2^{p'}$ words are within a Levenshtein distance of $2p'$.
In other words, $D_t(x)$ contains $2^{p'}$ words $\{z_1, \dots, z_{2^{p'}}\}$
such that all are subsequences of $y \in \Sigma^{n - t + p'}$.
We have
\begin{align*}
  \left|\bigcup_{i = 1}^{2^{p'}} I_t(z_i)\right|
  & \leq |I_t(z_1)| + \left|\bigcup_{i = 2}^{2^{p'}} (I_t(z_i) \setminus I_{t-p'}(y))\right|\\
  & \leq I_q(n-t, t) + (2^{p'} - 1) (I_q(n-t,t) - I_q(n-t+p', t-p'))\\
  & = 2^{p'} \cdot I_q(n-t,t) - (2^{p'} - 1) \cdot I_q(n-t+p', t-p').
\end{align*}
Let $Z = \{z_1, \dots, z_{2^{p'}}\}$.
For any integer $p' \leq p(x)$ with $p' \leq \min\{t,n-t\}$, it holds that
\begin{align*}
  |L_{t,t}(x)| & = \left| \bigcup_{z \in D_t(x)}I_t(z) \right| \\
  & \leq \sum_{z \in D_t(x) \setminus Z} |I_t(z)| +  \left| \bigcup_{z \in Z} I_t(z) \right| \\
  & \leq  (|D_t(x)| - 2^{p'}) \cdot I_q(n-t,t) +2^{p'} \cdot I_q(n-t,t) - (2^{p'} - 1) \cdot I_q(n-t+p', t-p')\\
  & = |D_t(x)| \cdot I_q(n-t,t) - (2^{p'} - 1) \cdot I_q(n-t+p', t-p') \triangleq L^*_{t,p'}(x).
\end{align*}

Note that $p'$ should satisfy $p' \leq \min\{ p(x), t, n-t\}$.
Since  $x \in \Sigma^n$ may have $p(x) = 0$,
it is impossible to choose $p' > 0$ such that $p' \leq p(x)$ for every $x \in \Sigma^n$.
Thus, we consider a subset $X \subseteq \Sigma^n$ and its partition $X = X_0 \cup X_1$ such that
$X_1 = \{ x \in X : p(x) \geq p' \}$ for some $p' \in [0, \min\{t,n-t\}]$ with $p' > 0$.
Then, as an upper bound on $|L_{t,t}(x)|$,
we can use $L^*_{t,p'}(x)$ for $x \in X_1$ and $|D_t(x)|\cdot I_q(n-t,t)$ for $x \in X_0$.

For integers $n$ and $p \in [0, n/2 ]$,
let $N_{n,q}(p)$ be the number of words $x \in \Sigma^n$ such that $p(x) = p$, where $|\Sigma| = q$.
The value of $N_{n,q}(p)$ is given by Theorem~\ref{thm:npq} below.
When $X_1$ is given by $\{ x \in \Sigma^n : p(x) \geq p' \}$,
it holds that $|X_1| = \sum_{p=p'}^{\lfloor n/2 \rfloor} N_{n,q}(p)$.

For a partition $X = X_0 \cup X_1$ with $X_1 = \{ x \in \Sigma^n : p(x) \geq p' \}$,
we have that
\begin{align}
  \sum_{x \in X} |L_{t,t}(x)| 
  & \leq \sum_{x \in X} L^*_{t,p'}(x) \nonumber\\
  & \leq \sum_{x \in X} |D_t(x)| \cdot I_q(n-t,t) - \sum_{x \in X_1}(2^{p'} - 1) \cdot I_q(n-t+p', t-p') \nonumber\\
  & = \sum_{x \in X} |D_t(x)| \cdot I_q(n-t,t) - (2^{p'} - 1) \cdot I_q(n-t+p', t-p')\sum_{p=p'}^{\lfloor n/2 \rfloor} N_{n,q}(p).\label{eq:Lupper}
\end{align}
We use (\ref{eq:Lupper}) for our lower bound of Theorem~\ref{thm:newlowerbound}.

We determine $N_{n,q}(p)$ for $p \in [0, n/2]$.

\begin{theorem}\label{thm:npq}
  For integers $n > 0$, $p \in [0, n/2]$, and $q \geq 2$, 
  \[
  N_{n,q}(p) = \binom{n-p}{p}\frac{q^p(q-1)^p}{n-p}(p + q(n-2p)).
  \]
\end{theorem}
\begin{proof}
We count the number of words $x \in \Sigma^n$ such that $p(x) = p$.
When $p(x) = 0$, $x$ should be a repetition of the same symbol in $\Sigma$.
Thus, $N_{n,q}(0) = q$.
Hereafter, we assume $p > 0$.
Let $(i_1, i_1+1), \dots, (i_p, i_p+1) \in \{1, \dots, n\}^2$ be $p$ index-pairs of distinct adjacent pairs such that $i_1 \leq \cdots \leq i_p$.
Let $(a_1, a_1'), \dots, (a_p, a_p') \in \Sigma^2$ be the corresponding distinct adjacent pairs,
where their concrete values are yet determined.
Since the order of $p$ pairs $(a_i, a_i')$ is fixed, we can construct all words $x$ with $p(x) = p$
by inserting $n-2p$ symbols into the word $a_1a_1'a_2a_2'\cdots a_pa_p'$.
There are $p+1$ possible places to which symbols can be inserted.
Namely, the resulting word should be of the form $w_1 a_1a_1' w_2 a_2a_2' w_3 \cdots w_p a_pa_p'w_{p+1}$, where $w_i \in \Sigma^*$.
Here, we consider the leftmost maximum-sized disjoint index-pair set in $x$.
By the property of distinct adjacent pairs, $w_1$ must be the empty string or repetition of $a_1$.
Similarly, for $i \in\{ 2, \dots, p\}$, $w_i$ must be the empty string or repetition of $a_i$.
Since we consider the \emph{leftmost} disjoint index-pairs,
different from the previous cases, $w_{p+1}$ can be the empty string or repetition of \emph{any} fixed symbol in $\Sigma$.

First, suppose that $w_{p+1}$ is the empty string.
Then the number of words $x$ with $p(x) = p$ is determined by
the number of possible lengths of $w_1, \dots, w_p$ and possible combinations $(a_i,a_i')$.
Since the rightmost symbols are fixed to be $a_pa_p'$,
the number of possible lengths is equal to
the the number of ways of selecting $p-1$ items from $n - 2p + (p-1)$ items, which is $\binom{n -  p-1}{p-1}$.
There are $q(q-1)$ combinations for each $(a_i, a_i')$.
Hence, the number $N_{n,q}(p)$ when $|w_{p+1}|= 0$ is equal to $q^p(q-1)^p \binom{n-p-1}{p-1}$.
Second, consider the case that $|w_{p+1}|  > 0$.
In this case, we count the number of possible lengths of $w_1, \dots, w_p, w_{p+1}$ and possible combinations $(a_i,a_i')$.
Since $|w_{p+1}| > 0$, we can assume that one symbol (length) has been assigned to $w_{p+1}$.
Then, the number of possible lengths is equal to
the number of ways of selecting $p$ items from $n - 2p - 1 + p$ items, which is $\binom{n- p -1}{p}$.
As in the previous case, the number of combinations for $(a_i, a_i')$'s are $(q(q-1))^p$,
and there are $q$ possible symbols corresponding to $w_{p+1}$.
Thus, the number $N_{n,q}(p)$ when $|w_{p+1}| > 0$ is $q^{p+1}(q-1)^p \binom{n-p-1}{p}$.

Therefore, for $p \geq 1$, we have
\begin{align*}
  N_{n,q}(p)
  & = q^p(q-1)^p\binom{n-p-1}{p-1} + q^{p+1}(q-1)^p \binom{n-p-1}{p}\\
  & = q^p(q-1)^p\binom{n-p}{p}\frac{p}{n-p} + q^{p+1}(q-1)^p \binom{n-p}{p}\frac{n-2p}{n-p}\\
  & = \binom{n-p}{p}\frac{q^p(q-1)^p}{n-p}(p + q(n-2p)).
\end{align*}
When $p=0$, the above gives $N_{n,q}(0) = q$. Hence, the statement follows.
\end{proof}

By using (\ref{eq:avelower})  and (\ref{eq:Lupper}), we have the following theorem.

\begin{theorem}\label{thm:newlowerbound}
  For a set $X \subseteq \Sigma^n$ and its partition $X = X_0 \cup X_1$ such that
  $X_1 = \{ x \in \Sigma^n : p(x) \geq p' \}$ for some $p' \in [0, \min\{t,n-t\}]$,
  it holds that
  \[A_q(n,d) \geq \left\lfloor \frac{|X|^2}{\sum_{x \in X} |D_t(x)| \cdot I_q(n-t,t) -(2^{p'} - 1) \cdot I_q(n-t+p', t-p') \sum_{p=p'}^{\lfloor n/2 \rfloor} N_{n,q}(p)}\right\rfloor,\]
  where $t = d/2 - 1$. 
\end{theorem}

We give an explicit expression by choosing $X = \Sigma^n$ and $p'=1$, 
\begin{corollary}\label{cor:ip_lb}
\[A_q(n,d) \geq \left\lfloor \frac{q^{n}}{q^{-t} \cdot I_q(n-t,t)^2
    - (1 - q^{-n+1}) I_q(n-t+1, t-1)}\right\rfloor,\]
  where $t = d/2 - 1 > 0$.
\end{corollary}
\begin{proof}
  When $X = \Sigma^n$, it follows from (\ref{eq:leveq}) that $\sum_{x \in X} |D_t(x)| \cdot I_q(n-t,t)  = q^{n-t} \cdot I_q(n-t,t)^2$.
  By choose $p'=1$, since $N_{n,q}(0)=q$ and $\sum_{p=0}^{\lfloor n/2 \rfloor}N_{n,q}(p) = q^n$, we have
  \begin{align*}
   (2^{p'} - 1) \cdot I_q(n-t+p', t-p') \sum_{p=p'}^{\lfloor n/2 \rfloor} N_{n,q}(p) 
    & = (q^n - q) \cdot I_q(n-t+1, t -1).
  \end{align*}
  By combining the above and Theorem~\ref{thm:newlowerbound}, the statement follows.
\end{proof}
Since Levenshtein~\cite{Lev02} gave a lower bound of $A_q(n,d) \geq q^n/(q^{-t}\cdot I_q(n-t,t)^2)$ for $t = d/2-1$,
Corollary~\ref{cor:ip_lb} improves his lower bound.
Its effect, however, vanishes when $d$ becomes large; Thus, the asymptotic lower bound on the coding rate given by Corollary~\ref{cor:ip_lb}
is the same as~\cite{Lev02}.

\section{Comparison}

First, we compare code-size upper bounds given by~\cite{Lev02,KK13} and Theorems~\ref{thm:spbound} and~\ref{thm:bound}.
Table~\ref{tb:upper_bounds} displays numerical values for several parameters $q$, $n$, and $d$.
Specifically, we use the bounds in \cite[Theorem~2]{Lev02}, where we take the minimum over parameter $r$,
and \cite[Corollary~4.2]{KK13}.
Though the non-asymptotic bound of~\cite{KK13} achieves the best when $d$ is small such as $4$ and $10$,
Theorem~\ref{thm:bound} is much better when $d$ is relatively large.

\begin{table*}[tbp]
  \caption{Upper Bounds on Code Size}\label{tb:upper_bounds}
  \medskip  \centering
  \begin{tabular}{ccc|rrrr}
    $q$ & $n$ & $d$ & UB of~\cite{Lev02} & UB of~\cite{KK13} & Theorem~\ref{thm:spbound} & Theorem~\ref{thm:bound}\\\hline
    
    2 & 20 & 4 & 97 453 & 55 206 & 95 325 & 181 643 \\
    2 & 20 & 10 & 26 456 & 2 535& 1 295 & 2 452 \\
    2 & 20 & 20 & 190 416 & 1 059 &  32 & 28 \\
    2 & 20 & 30 & 961 048 & --- & 5 & 4 \\

    2 & 40 & 4 & 47 498 012 376 & 28 192 605 878 & 52 357 696 560 & 102 167 009 660 \\
    2 & 40 & 10 & 1 279 636 864 & 9 880 934 & 117 292 187 & 228 473 245 \\
    2 & 40 & 20 & 1 122 371 648 & 3 203 459 & 215 900 & 203 859 \\

    2 & 40 & 30 & 13 097 807 352 & 298 539 & 3 735 & 1 195 \\
    2 & 40 & 40 & 287 193 094 240 & 1 048 713 & 231 & 43 \\

    
    4 & 20 & 4 & 30 003 945 118 & 19 289 677 788 & 68 719 476 736 & 90 174 299 388 \\
    4 & 20 & 10 & 360 221 648 & 25 002 768 & 306 647 351 & 316 287 316 \\
    4 & 20 & 20 & 536 774 720 & 1 645 315 & 1 258 226 &  79 926\\
    4 & 20 & 30 & 192 278 071 952 & --- & 34 771 &  71\\

    4 & 40 & 4 & $\approx$ 14 843 $\times 10^{18}$ & $\approx$ 10 332 $\times 10^{18}$ & $\approx$ 38 997 $\times 10^{18}$& $\approx$ 51 574  $\times 10^{18}$\\    
    4 & 40 & 10 & $\approx$ 6 113 $\times 10^{15}$ & $\approx$ 461 805 $\times 10^{12}$  & $\approx$ 27 238 $\times 10^{15}$ &            $\approx$ 36 015 $\times 10^{15}$ \\
    4 & 40 & 20 & $\approx$ 133 526 $\times 10^{12}$ & $\approx$ 158 374 $\times 10^{9}$ & $\approx$ 7 269 $\times 10^{12}$ &                       $\approx$ 2 172 $\times 10^{12}$ \\

    4 & 40 & 40 &      $\approx$ 34 641 $\times 10^{15}$ & $\approx$ 2 123 $\times 10^{9}$ &$\approx$ 173 431 $\times 10^{6}$ &                         $\approx$   69 $\times 10^{6}$ \\
    4 & 40 & 60 &  $\approx$ 306 026 $\times 10^{18}$ & ---  & 164 423 496 &                       108 \\

  \end{tabular}
\end{table*}

\begin{figure*}[ht]
    \centering
    \includegraphics[width=0.8\textwidth,pagebox=cropbox]{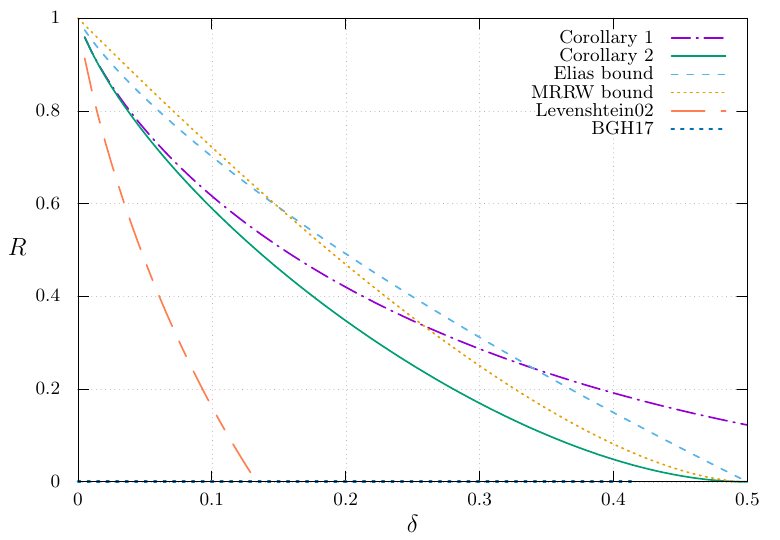}
    \caption{Bounds for $q=2$}
    \label{fig:bound_qis2}
\end{figure*}
\begin{figure*}[ht]
    \centering
    \includegraphics[width=0.8\textwidth,pagebox=cropbox]{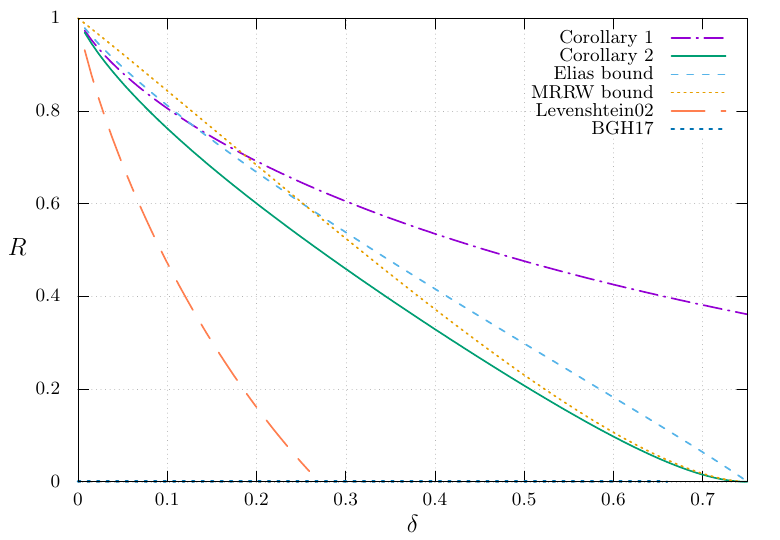}
    \caption{Bounds for $q=4$}
    \label{fig:bound_qis4}
\end{figure*}

Next, we investigate asymptotic behaviors.
Figures~\ref{fig:bound_qis2} and~\ref{fig:bound_qis4} show asymptotic bounds on  the rate function $R_q(\delta)$ for $q=2$ and $q=4$, respectively.
Corollary~\ref{cor:ub} gives the best upper bounds in both cases.
For larger $q$, Corollary~\ref{cor:ub} is inferior to the MRRW bound, especially for large $\delta$.
Regarding lower bounds, \cite{Lev02} gave a bound (displayed as Levenshtein02),
and \cite{BGH17} showed that there exist $q$-ary codes of rate $R$ that achieve
$\delta \geq 1 - (2/(q+ \sqrt{q}))$ for sufficiently small $R >0$ (displayed as BGH17).
Also, it is proved in~\cite{GHL21} that there is a small constant $\epsilon >0$ such that $R = 0$ for $\delta \geq 1 - (1+\epsilon)/q$
(omitted in the figures).
The asymptotic upper bounds of~\cite{Lev02,KK13} were not given in the closed form.
Hence it is not easy to make a clear comparison. 
By comparing with the plotted bound of~\cite[Fig.~1]{KK13}, we can see that Corollary~\ref{cor:ub} gives tighter bounds
on $\delta \geq 0.1$ for $q=2$ and $\delta \geq 0.2$ for $q=4$.

\begin{table*}[tbp]
  \caption{Lower Bounds on Code Size}\label{tb:lower_bounds}
  \medskip  \centering
  \begin{tabular}{ccc|rr}
    $q$ & $n$ & $d$ & \multicolumn{1}{r}{LB of~\cite{Lev02}}  & \multicolumn{1}{r}{Corollary~\ref{cor:ip_lb}} \\\hline
    2 & 20 & 6 & 94   & 94 \\
    2 & 20 & 8 & 4   & 4 \\
    2 & 40 & 6 & 6 524 894  & 6 526 482  \\
    2 & 40 & 8 & 76 814    & 76 818 \\
    2 & 40 & 10 & 1 687   & 1 687 \\
    4 & 20 & 6 & 5 608 964  & 5 610 710 \\
    4 & 20 & 8 & 66 412   & 66 419 \\

    4 & 40 & 6 & 379 316 355 894 427 152 & 379 330 757 315 377 297 \\
    4 & 40 & 8 & 1 031 317 510 055 795 &  1 031 323 792 762 824 \\
      4 & 40 & 10 & 5 251 871 945 006  & 5 251 878 194 182 
    \end{tabular}
\end{table*}

Table~\ref{tb:lower_bounds} shows the numerical values of the lower bounds of~\cite[Theorem~1]{Lev02} 
and Corollary~\ref{cor:ip_lb}.
As we can see, the improvement by Corollary~\ref{cor:ip_lb} is small.

\section{Conclusions}

This paper has presented improved upper and lower bounds on code size for correcting insertions and deletions.
In particular, our upper bound improves the existing bounds in both non-asymptotic and asymptotic senses.
An interesting future work is to develop an upper bound superior to the MRRW bound for large alphabet size $q$
and relative distance $\delta$.
Our bound is inferior in some range for $q \geq 5$.
A key lemma may be a list-size upper bound of insertions/deletion (as Lemma~\ref{lem:hy20}) employing $q$ explicitly in the bound.





\section*{Acknowledgment}
This work was supported in part by JSPS Grants-in-Aid for Scientific Research Number 18K11159.

\bibliographystyle{abbrv}
\bibliography{mybib}

\end{document}